\documentclass[preprint,12pt]{elsarticle}

\usepackage{graphicx}
\usepackage{mathptmx}  
\usepackage[a4paper,left=35mm,top=26mm,right=26mm,bottom=18mm]{geometry}
\usepackage{amsmath}
\usepackage{amssymb,amsfonts,amsthm}
\usepackage[all,arc]{xy}
\usepackage{enumerate}
\usepackage{mathrsfs}
\usepackage{float}
\usepackage{hyperref}
\usepackage{textcomp}
\usepackage{siunitx}
\usepackage{textcomp}
\usepackage{pgf,tikz}
\usepackage{mathrsfs}
\usetikzlibrary{arrows}
\newcommand{\w}{\omega}
\newtheorem{theorem}{Theorem}

\journal{Graphs and Combinatorics}

\makeatletter
\def\@author#1{\g@addto@macro\elsauthors{\normalsize%
    \def\baselinestretch{1}%
    \upshape\authorsep#1\unskip\textsuperscript{%
      \ifx\@fnmark\@empty\else\unskip\sep\@fnmark\let\sep=,\fi
      \ifx\@corref\@empty\else\unskip\sep\@corref\let\sep=,\fi
      }%
    \def\authorsep{\unskip,\space}%
    \global\let\@fnmark\@empty
    \global\let\@corref\@empty  
    \global\let\sep\@empty}%
    \@eadauthor={#1}
}
\makeatother

\begin{document}

\begin{frontmatter}


\title{Colouring of $(P_3 \cup P_2)$-free graphs}

 \author{Arpitha P. Bharathi\fnref{label1}\corref{cor1}}
 \ead{arpitha.p.bharathi@gmail.com}
 \author{Sheshayya A. Choudum\fnref{label1}}
 \ead{sac@retiree.iitm.ac.in}
 \fntext[label1]{Department of Mathematics, Christ University, Bengaluru 560029, India}
 \cortext[cor1]{Corresponding author}



\begin{abstract}
The class of $2K_{2}$-free graphs and its various subclasses have been studied in a variety of contexts. In this paper, we are concerned with the colouring of $(P_{3}\cup P_{2})$-free graphs, a super class of $2K_{2}$-free graphs. We derive a $O(\omega^{3})$ upper bound for the chromatic number of $(P_{3} \cup P_{2})$-free graphs, and sharper bounds for $(P_{3} \cup P_{2}$, diamond)-free  graphs, where $\omega$ denotes the clique number. By applying similar proof techniques we obtain chromatic bounds for $(2K_{2},$ diamond)-free graphs. The last two classes are perfect if $\omega \geq 5$ and $\geq 4$ respectively.
\end{abstract}


\end{frontmatter}


\section{Introduction}\label{sec:intro}
A graph $G$ is said to be $H$-free, if $G$ does not contain an induced copy of $H$. More generally, a class of graphs $\mathcal{G}$ is said to be $(H_{1},H_{2},\cdots)$-free if every $G \in$ $\mathcal{G}$ is $H_{i}$-free, for $i \geq 1$.
The class of $2K_{2}$-free graphs and its subclasses are subject of research in various contexts: domination (El-Zahar and Erd\"{o}s \cite{ElZahar}), size (Bermond et al. \cite{Bermond}, 
Chung et al. \cite{Chung}), vertex colouring (Wagon \cite{Wagon}, Nagy and Szentmiklossy \cite{ENS}, Gy\'{a}rf\'{a}s \cite{Gyarfas}), edge colouring (Erd\"{o}s and Nesetril \cite{ErdosNesetril}) and algorithmic complexity (Blazsik et al. \cite{Blazsik}). Here we are concerned with the colouring of $(P_{3} \cup P_{2})$-free graphs, a super class of $2K_{2}$-free graphs. A comprehensive result of Kral et al. \cite{Kral} states that the decision problem of COLOURING $H$-free graphs is P-time solvable if $H$ is an induced subgraph of $P_{4}$ or $P_{3} \cup P_{1}$,and it is NP-complete for any other graph $H$. In particular, COLOURING $2K_{2}$-free graphs is NP-complete. However, there have been several studies to obtain tight upper and lower bounds for the chromatic number of $2K_{2}$-graphs. A problem of Gy\'{a}rf\'{a}s \cite{Gyarfas} asks for the smallest function $f(x)$ such that $\chi(G) \leq f(\omega(G)$, for every $G$ belonging to the class of $2K_{2}$-free graphs, where $\chi(G)$ and $\omega(G)$ respectively denote the chromatic number and clique number of $G$. This problem is still open. In this respect, an often quoted result is due to Wagon \cite{Wagon}. It states that if a graph $G$ is $2K_{2}$-free, then $\chi(G)\leq \binom{\omega(G)+1}{2}$. We look more closely at Wagon's proof and obtain a $O (\omega^{3})$  upper bound for the chromatic number of $(P_{3} \cup P_{2})$-free graphs, and  sharper bounds for $(P_{3} \cup P_{2}$, diamond)-free  graphs.  By applying similar proof techniques we obtain chromatic bounds for $(2K_{2},$ diamond)-free graphs. The last  two classes are perfect if the clique number is $\geq 5$ and $\geq 4$ respectively. The classes of $(H,$ diamond)-free graphs and $(H_1,H_2,$ diamond)-free graphs, for various graphs $H,H_1$ and $H_2$, have been studied in many papers; see Arbib and Mosca \cite{Arbib}, Brandst\"{a}dt \cite{Brandstadt}, Choudum and Karthick \cite{SAC}, Karthick and Maffrey \cite{Karthick}, Gy\'{a}rf\'{a}s \cite{Gyarfas}, and Randerath and Schiermeyer \cite{Rand}. See also a comprehensive book on problems of graph colourings by Jensen and Toft \cite{Jensen} and an extensive book of Brandtst\"{a}dt et al. \cite{Brandstadtetal}, for interesting subclasses and superclasses of $2K_{2}$-free graphs.

\section{Terminology and Notation}\label{sec:term}
We follow standard terminology of Bondy and Murty \cite{Bondy}, and West \cite{West}. All our graphs are simple and undirected. If  $u$, $v$ are two vertices of a graph $G(V,E)$, then their adjacency is denoted by $u\leftrightarrow v$, and non-adjacency by $u\nleftrightarrow v$. $P_{n}, C_{n}$ and $K_{n}$ respectively denote the path, cycle and complete graph on $n$ vertices. A chordless cycle of length $\geq 5$ is called a \emph{hole}. If $S \subseteq V(G)$, then $[S]$ denotes the subgraph induced by $S$. If $S$ and $T$ are two disjoint subsets of $V(G)$, then $[S,T]$ denotes the set of edges $\{st\in E(G): s \in S$ and $t \in T \}$. A subset $Q$ of $V(G)$ is called a \textit{clique} if  any two vertices in $Q$ are adjacent. The \textit{clique number} of $G$ is defined to be max$\{|Q|: Q$ is a clique in $G\}$; it is denoted by $\omega(G)$. A clique $Q$ is called a \textit{maximum clique} if $|Q| = \omega(G))$. A subset $I$ of $V(G)$ is called an \textit{independent set} if no two vertices in $I$ are  adjacent. A \textit{k-vertex colouring} or a \emph{k-colouring} or a \emph{colouring} is a function $f: V(G) \rightarrow \{1,2,\cdots ,k\}$ such that $f(u) \neq f(v)$, for any two adjacent vertices $u$, $v$ in $G$. It is also referred to as a proper colouring of $G$ for emphasis. The \textit{chromatic number} $\chi(G)$ of $G$ is defined to be min$\{k: G$ admits a \emph{k}-colouring$\}$. If $G_{1},G_{2},\cdots ,G_{k}$ are vertex disjoint graphs, then  $G_{1}\cup G_{2}\cup \cdots \cup G_{k}$ denotes the graph with vertex set $\bigcup_{i=1}^{k}V(G_{i})$ and edge set  $\bigcup_{i=1}^{k}E(G_{i})$. If $G_{1}\simeq G_{2} \simeq \cdots \simeq G_{k}\simeq H$, for some $H$, then $G_{1}\cup G_{2}\cup \cdots \cup G_{k}$ is denoted by $kH$.  The three graphs which appear frequently in this paper are shown in Fig.1.

\begin{figure}[!htb]
\minipage{0.32\textwidth}
		\scalebox{1.3}{
		\begin{tikzpicture}[line cap=round,line join=round,>=triangle 45,x=1.0cm,y=1.0cm]
\clip(1.5,1.5) rectangle (3.5,3.5);
\draw (2.,3.)-- (3.,3.);
\draw (2.,2.)-- (3.,2.);
\begin{scriptsize}
\draw [fill=black] (2.,3.) circle (1.2pt);
\draw [fill=black] (2.,2.) circle (1.2pt);
\draw [fill=black] (3.,3.) circle (1.2pt);
\draw [fill=black] (3.,2.) circle (1.2pt);
\end{scriptsize}
\end{tikzpicture}

		}
\endminipage\hfill
\minipage{0.32\textwidth}
		\scalebox{1.3}{
		\begin{tikzpicture}[line cap=round,line join=round,>=triangle 45,x=1.0cm,y=1.0cm]
\clip(1.5,1.5) rectangle (4.5,3.5);
\draw (2.,3.)-- (3.,3.);
\draw (3.,3.)-- (4.,3.);
\draw (2.5,2.)-- (3.5,2.);
\begin{scriptsize}
\draw [fill=black] (2,3) circle (1.2pt);
\draw [fill=black] (3,3) circle (1.2pt);
\draw [fill=black] (4,3) circle (1.2pt);
\draw [fill=black] (2.5,2) circle (1.2pt);
\draw [fill=black] (3.5,2) circle (1.2pt);
\end{scriptsize}
\end{tikzpicture}

    }
\endminipage\hfill
\minipage{0.32\textwidth}
	\centering
		\scalebox{.6}{
		\begin{tikzpicture}[line cap=round,line join=round,>=triangle 45,x=1.0cm,y=1.0cm]
\clip(5.5,-1.5) rectangle (8.5,3.5);
\draw (6.,1.)-- (8.,1.);
\draw (7.02,2.76)-- (6.,1.);
\draw (7.02,2.76)-- (8.,1.);
\draw (7.02,-0.78)-- (8.,1.);
\draw (7.02,-0.78)-- (6.,1.);
\begin{scriptsize}
\draw [fill=black] (6.,1.) circle (2.5pt);
\draw [fill=black] (8.,1.) circle (2.5pt);
\draw [fill=black] (7.02,2.76) circle (2.5pt);
\draw [fill=black] (7.02,-0.78) circle (2.5pt);
\end{scriptsize}
\end{tikzpicture}
   }
\endminipage\hfill
\caption{$2K_2$, $P_3 \cup P_2$, Diamond}
\end{figure}
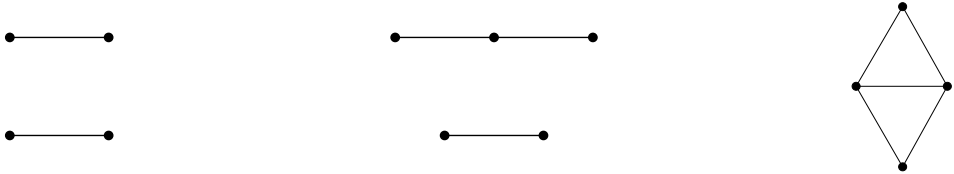

\section{A partition of the vertex set of a graph.}

Throughout this paper we use a particular partition of the vertex set of a graph $G(V,E)$ and use its properties. Some of these properties are due to Wagon \cite{Wagon}, but  are restated for ready reference. In what follows, $\w$ denotes the clique number of a graph  under consideration.\\
\\Let $A$ be a maximum clique in $G$ with vertices $1, 2,\cdots ,\omega$. We iteratively define the sets $C_{ij}$  in the lexicographic order of pairs of vertices $i$, $j$ of $A$.\\
\\$C=\phi$\\
for $i$ : 1 to $\w$\\
$\;\;\;$    for $j$ : $i+1$ to $\w$\\
$\qquad    C_{ij}=\{v\in V(G)-C \: | \: v \nleftrightarrow i$ and $v \nleftrightarrow j\};$\\
$\qquad    C=C\cup C_{ij};$\\
$\quad    $end\\
end\\
\\

By definition, there are  $\binom{\omega}{2}$  number of $C_{ij}$ sets and these are pairwise disjoint. Also, every vertex in $C_{ij}$ is adjacent to every vertex $k$ of $A$, where $1\leq k < j, k \neq i$. Moreover, every vertex in $V(G) - A$ which is  non-adjacent to  two or more vertices of $A$ is in some $C_{ij}$. So, every vertex $v \in V(G)- (A \cup C)$ is adjacent to all the vertices of $A$ or $|A| - 1$  vertices  of $A$. The former case is impossible, since $A$ is a maximum clique. Hence we define the following sets. For $a \in A$, let
\begin{center}
$I_a=\{v\in V(G) - (A \cup C)\;\mid\;v\leftrightarrow\ x , \forall x\in A-\{a\}$ and $ v\nleftrightarrow a\}$.
\end{center}

By the above remarks, we conclude that $(A ,\bigcup\limits_{i,j} C_{ij}, \bigcup\limits_{a\in A} I_{a})$ is a partition of $V(G)$.

\section{Colouring of \texorpdfstring{$(P_3 \cup P_2)$}{Lg}-free graphs}

We first observe a few properties of the sets $C_{ij}$ and $I_{a}$, and then obtain an $O(\omega^{3})$ upper bound for the chromatic number of a $(P_3 \cup P_2)$-free graph.

\begin{theorem}\label{thm:P3UP2}
If a graph G is $(P_3 \cup P_2)$-free, then $\chi(G)\leq \frac{\omega(\omega+1)(\omega+2)}{6}$.
\end{theorem}

\begin{proof} Let \emph{A} be a maximum clique in \emph{G}. Let $(1, 2, 3,\cdots,\omega)$ be a vertex ordering of \emph{A}. Since $G$ is $(P_3 \cup P_2)$-free,  the sets $C_{ij}$ and $I_{a}$ possess a few more properties, in addition to the ones stated in section 3.\\
\\ \emph{Claim 1: Each induced subgraph $[C_{ij}]$ of $G$ is $P_{3}$-free and hence it is a disjoint union of cliques.}\\
If some $C_{ij}$ contains an induced $P_3=(x,y,z)$, then $[\{x,y,z\}\cup\{ i, j\}] \simeq P_3 \cup P_2$, a contradiction.\\
\\ \emph{Claim 2: Each $I_{a}$ is an independent set.}\\
If some $I_{a}$ contains an edge $vw$, then $A \cup \{v,w\} - \{a\}$ is a clique of size $\omega+1$, a contradiction to the maximality of $|A|$.\\
\\ \emph{Claim 3: $\omega\left([C_{ij}]\right)\leq \omega -(j-2)$, where $j\geq 2$}\\
\noindent Let $B$ be a maximum clique in $[C_{ij}]$. Every vertex in $B$ is adjacent to every vertex in $K = \{1,2,\cdots, j-1\} - \{i\} \subseteq A$, by the definition of $C_{ij}$. So, $B \cup K$ is a clique of $G$. Hence, $\omega(G) \geq |B \cup K| = \omega([C_{ij}]) + |K| = \omega([C_{ij}]) + j-2$. Hence the claim.\\
\\Table \ref{table} gives the the number of sets $C_{ij}$, for a fixed $j$, where $i < j$ and $2\leq j \leq \omega$. The entries of the last column, follow by Claim 3.

\begin{table}
\caption{Clique size of each $[C_{ij}]$}
\label{table}       

\centering
\begin{tabular}{llll}
\hline\noalign{\smallskip}
$j$ & $C_{ij}$'s & Number of $C_{ij}$'s & $\omega([C_{ij}])\leq$ \\
\noalign{\smallskip}\hline\noalign{\smallskip}
2 & $C_{12}$ & 1 & $\omega$ \\
3 & $C_{13},C_{23}$ & 2 & $\omega-1$ \\
4 & $C_{14},C_{24},C_{34}$ & 3 & $\omega-2$ \\
. & . & . & . \\
. & . & . & . \\
. & . & . & . \\
$j$ & $C_{1j},C_{2j},...,C_{j-1\;j}$ &  $j-1$ & $\omega-(j-2)$ \\
. & . & . & . \\
. & . & . & . \\
. & . & . & . \\
$\omega$ & $C_{1\;\omega},C_{2\;\omega},...,C_{\omega-1\;\omega}$ & $\omega - 1$ & 2 \\
\noalign{\smallskip}\hline
\end{tabular}

\end{table}

We now properly colour $G$ as follows:
\begin{enumerate}[(1)]
\item Colour the vertices $1,2,\cdots, \omega$  of $A$ with colours $1,2,\cdots, \omega$ respectively.
\item Colour the vertices of $C_{ij}$ with $\omega([C_{ij}])$ new colours, $1 \leq i < j \leq \omega$.  By Claim 1, $[C_{ij}]$ is a disjoint union of cliques and hence one can properly colour $[C_{ij}]$ with $\omega([C_{ij}])$ colours. Note also that one requires at most $\omega - (j-2)$ colours, by Claim 3.
\item Each vertex in $I_{a}$ is given the colour of $a \in A$.
\end{enumerate}
It is a proper colouring of $G$ by Claims 1, 2 and 3. We first estimate the number of colours used in step (2) to colour the vertices of $C$ (see Table 1) and then estimate the total number of colours used to colour $G$ entirely.

\begin{align*}
\chi([C]) &\leq 1(\omega) + 2(\omega -1) + 3(\omega-2) + ... + (\omega-1)2\\
&= \sum_{k=1}^{\omega-1}{k(\omega + 1 - k)}\\
&= \sum_{k=1}^{\omega-1}{k(\omega + 1)}-\sum_{k=1}^{\omega-1}{k^2}\\
&= (\omega + 1)\frac{(\omega -1)(\omega)}{2} - \frac{(\omega-1)(\omega)(2\omega-2+1)}{6}\\
&=  \frac{\omega(\omega-1)(\omega+4)}{6}
\end{align*}
\\
Hence,
\begin{align*}
\chi(G) &\leq |A| + \chi([C])\\
&= \omega + \frac{\omega(\omega-1)(\omega+4)}{6}\\
& = \frac{\omega(\omega+1)(\omega+2)}{6}
\end{align*}
\end{proof}

\begin{theorem}\label{thm:P4UP2}
If a graph $G$ is $(P_4 \cup P_2)$-free,  then $\chi(G)\leq \frac{\omega(\omega+1)(\omega+2)}{6}$.
\end{theorem}

\begin{proof}
The bound for the chromatic number of $(P_3 \cup P_2)$-free graphs holds for $(P_4 \cup P_2)$-free graphs too. In this case, each $[C_{ij}]$ is  $P_{4}$-free and hence perfect, by a result of Seinsche \cite{Seinsche}. So, we can properly colour each $[C_{ij}]$ with at most  $\omega(C_{ij}) \leq \omega - (j-2)$ colours, and the entire $G$ with at most $\frac{\omega(\omega+1)(\omega+2)}{6}$ colours , as in the proof of Theorem \ref{thm:P3UP2}.
\end{proof}

We next consider $(P_3 \cup P_2,$ diamond$)$-free graphs and obtain sharper bounds for the chromatic number. If $\omega = 1$, then obviously chromatic number is 1. So in the following, all graphs have $\omega \geq 2$.

\begin{theorem} \label{thm:P3UP2,diamond}
If a graph G is  $(P_3 \cup P_2,$ diamond$)$-free,  then
\begin{center}
\begin{math}
\chi(G) \leq \begin{cases} \w+2 & \mbox{if } \w=2 \\
\w+3 & \mbox{if } \w=3\\
\w+1 & \mbox{if } \w=4
\end{cases}
\end{math}
\end{center}
and $G$ is perfect if $\w \geq 5.$
\end{theorem}

\begin{proof} We continue to use the terminology and notation of sections 2 and 3. In particular, we use the sets $A$, $C_{ij}$, $I_{a}$, and Claims 1, 2 and 3.\\
\\ \emph{Claim 4: If $G$ is $C_{5}$-free, then it is a perfect graph.}\\
Clearly,  every hole $C_{2k+1}, k \geq 3$ contains an induced $P_3 \cup P_2$, and the complement $\overline{C}_{2k+1}, k \geq 3$ of the hole contains an induced diamond. So $G$ is $(C_{2k+1}, \overline{C}_{2k+1})$-free for all $k \geq 3$. Hence if $G$ is  $C_{5}$-free, then  $G$ is perfect, by the Strong Perfect Graph Theorem \cite{SPGT}.\\
\\ \emph{Claim 5: $C_{ij} = \emptyset$, for every $j \geq 4$}.\\
On the contrary, let $x \in C_{ij}$, for some $j \geq 4$. Then by  the definition of $C_{ij}$, there exist   two distinct vertices $p,q \in \{1,2,3\} \subseteq A$ such that $x \leftrightarrow p$ and $x \leftrightarrow q$. But then $[\{x,j,p,q\}] \simeq$ diamond, a contradiction.\\
\\So, we conclude that  $C = C_{12} \cup C_{13} \cup C_{23}$, for $j \geq 4$.\\
\\ \emph{Claim 6: If $a \in A$, then $I_{a}$ is an empty set  if $\omega \geq 3$, and it is an independent set  if $\omega = 2$}.\\
If  $\omega \geq 3$, and  $x \in I_a$, for some $a \in A - \{1,2\}$, then $[ \{x,a,1,2\}] \simeq$  diamond, a contradiction; if a = 1 or 2, then
 [ \{x,1,2,3\}] is a diamond.  If  $\omega = 2$, then the assertion follows by Claim 2. \\
\\
Therefore, $V(G) = A \cup C_{12} \cup C_{13} \cup C_{23}$, if $\omega \geq 3$.
\\
\\
Recall that by Claim 3,  $\omega ([C_{13}]) \leq  \omega - 1$, and $\omega ([C_{23}]) \leq \omega - 1$. But  $[C_{12}]$ may contain an $\omega$-clique. However, we have the following claim.
\\
\\
\noindent \emph{Claim 7: $\omega ([C_{12}]) \leq  \omega - 1$, if $\omega(G) \geq 3$, and $C_{23} \neq \emptyset$ or $C_{13} \neq \emptyset$}  \\
\noindent On the contrary suppose $[C_{12}]$ contains an $\omega$-clique $Q$, and for definiteness suppose $C_{23} \neq \emptyset$ (if $C_{13} \neq \emptyset$, proof is similar). Let $x \in C_{23}$. If $x$ is adjacent to all the vertices of $Q$ or $|Q|-1$ vertices of $Q$, then we have an $(\omega + 1)$-clique or a diamond in $G$, both impossible. Else, there exist two vertices $u$ and $v$ in $Q$ such that $x \nleftrightarrow u$ and  $x \nleftrightarrow v$. Then $[\{x,1,2\}\cup \{u,v\}] \simeq P_{3}\cup P_{2}$, a contradiction. Hence the claim.
\\
\\
\noindent \emph{Claim 8: $[C_{13},\: A - \{2 \}] = \emptyset$, and $[C_{23}, A - \{1 \}] = \emptyset$.}\\
\noindent If there exists an edge $xi \in  [C_{13},\: A - \{2 \}]$, then $[\{x, i, 1,2\}] \simeq$ diamond, a contradiction. Similarly, $[C_{23},\: A - \{1 \}] = \emptyset$
\\
\\
We now prove the theorem for different values of $\omega$, by making the cases as stated in the theorem.
\\
\\
\textbullet \quad   $\mathbf \omega = 2$; so $A = \{1,2\}$.\\
\\
Colouring $G$ with four colours is easy in this case, since  $V(G) = A \cup C_{12} \cup I_{1}\cup I_{2}$, $\omega ([C_{12}]) \leq  \omega = 2$, and $I_{1}$, $I_{2}$ are independent sets, by Claim 6. Moreover, $\omega[C_{12}] \leq \omega(G) = 2$. The following is a proper 4-colouring of $G$:
\begin{enumerate}[(1)]
\item Colour the vertices 1 and 2 of $A$ with colours 1 and 2 respectively.
\item Colour $[I_{1}]$ with  colour 1.
\item Colour $[I_{2}]$ with  colour 2.
\item Colour $[C_{12}]$ with two new colours.
\end{enumerate}

An extremal $(P_{3} \cup P_{2},$ diamond)-free graph $G$ with $\omega(G) = 2$, and $\chi(G) =4$ is the Mycielski-Gr\"{o}tzsch graph; see Fig. \ref{fig:MGgraph}. It is well known that this graph has clique number 2 and chromatic number 4. The graph is clearly diamond free since it is triangle free. It can be observed that this graph is $(P_3 \cup P_2)$-free by selecting every edge $P_2$ and then verifying that the second neighborhood of $P_{2}$, is $P_{3}$-free. There are not too many cases for such a verification because of the symmetry of edges; we need to choose only three kinds of edges: $v_1v_2$, $v_1u_2$ and $u_1w$.
\\
\\
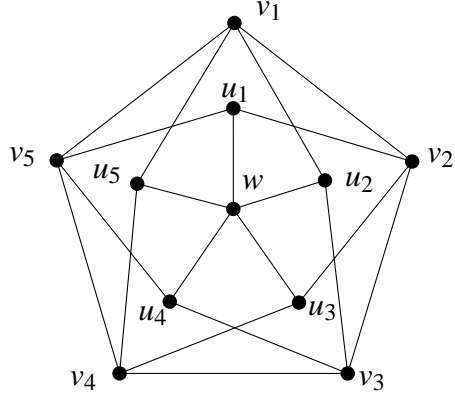
\begin{figure}
	\centering
		\scalebox{1}{
		\begin{tikzpicture}[line cap=round,line join=round,>=triangle 45,x=1.0cm,y=1.0cm]
\clip(-2.5,-2.5) rectangle (3.5,3.2);
\draw (-0.33984049107505293,-1.0473455666519236)-- (2.,-2.);
\draw (1.36,-1.06)-- (-1.,-2.);
\draw (1.36,-1.06)-- (0.4963897200781277,0.1809056989827751);
\draw (0.4963897200781277,0.1809056989827751)-- (-0.33984049107505293,-1.0473455666519236);
\draw (2.9,1.1336848906216752) node[anchor=north west] {$v_2$};
\draw (0.6431198639104309,3.0889125150551324) node[anchor=north west] {$v_1$};
\draw (1.9956081454734185,-1.82120798540182) node[anchor=north west] {$v_3$};
\draw (-1.8,-1.82120798540182) node[anchor=north west] {$v_4$};
\draw (-2.6,1.1483858502038817) node[anchor=north west] {$v_5$};
\draw (0.18739011686203277,2.0) node[anchor=north west] {$u_1$};
\draw (1.8191966304869418,0.8102637798131336) node[anchor=north west] {$u_2$};
\draw (1.3340649642741311,-0.8950475317228139) node[anchor=north west] {$u_3$};
\draw (-0.9,-0.9538513700516398) node[anchor=north west] {$u_4$};
\draw (-1.5,0.9278714564707851) node[anchor=north west] {$u_5$};
\draw (0.46670834892395413,0.7808618606487207) node[anchor=north west] {$w$};
\draw (0.4963897200781277,0.1809056989827751)-- (-0.7681722559813825,0.5162445881690046);
\draw (0.4963897200781277,0.1809056989827751)-- (0.49611026808836695,1.515909839759043);
\draw (0.4963897200781277,0.1809056989827751)-- (1.7015889538292908,0.560347466915624);
\draw (-0.7681722559813825,0.5162445881690046)-- (-1.,-2.);
\draw (1.7015889538292908,0.560347466915624)-- (2.,-2.);
\draw (-1.,-2.)-- (-1.8266413459002424,0.82496473939534);
\draw (-1.8266413459002424,0.82496473939534)-- (0.5108112276705733,2.647883727588939);
\draw (0.5108112276705733,2.647883727588939)-- (2.848263801241389,0.8102637798131336);
\draw (2.848263801241389,0.8102637798131336)-- (2.,-2.);
\draw (2.,-2.)-- (-1.,-2.);
\draw (-0.7681722559813825,0.5162445881690046)-- (0.5108112276705733,2.647883727588939);
\draw (1.7015889538292908,0.560347466915624)-- (0.5108112276705733,2.647883727588939);
\draw (-1.8266413459002424,0.82496473939534)-- (0.49611026808836695,1.515909839759043);
\draw (0.49611026808836695,1.515909839759043)-- (2.848263801241389,0.8102637798131336);
\draw (-0.33984049107505293,-1.0473455666519236)-- (-1.8266413459002424,0.82496473939534);
\draw (1.36,-1.06)-- (2.848263801241389,0.8102637798131336);
\begin{scriptsize}
\draw [fill=black] (-1.,-2.) circle (2.5pt);
\draw [fill=black] (2.,-2.) circle (2.5pt);
\draw [fill=black] (-0.33984049107505293,-1.0473455666519236) circle (2.5pt);
\draw [fill=black] (1.36,-1.06) circle (2.5pt);
\draw [fill=black] (0.4963897200781277,0.1809056989827751) circle (2.5pt);
\draw [fill=black] (-0.7681722559813825,0.5162445881690046) circle (2.5pt);
\draw [fill=black] (0.49611026808836695,1.515909839759043) circle (2.5pt);
\draw [fill=black] (1.7015889538292908,0.560347466915624) circle (2.5pt);
\draw [fill=black] (-1.8266413459002424,0.82496473939534) circle (2.5pt);
\draw [fill=black] (2.848263801241389,0.8102637798131336) circle (2.5pt);
\draw [fill=black] (0.5108112276705733,2.647883727588939) circle (2.5pt);
\end{scriptsize}
\end{tikzpicture}}
	\caption{Mycielski-Gr\"{o}tzsch graph}
	\label{fig:MGgraph}
\end{figure}

\noindent \textbullet \quad   $\mathbf \omega = 3$; so $A = \{1,2,3\}$.\\

At the outset, recall that every $I_{a} = \emptyset$, by Claim 6.  So, $V(G) = A \cup C_{12} \cup C_{23} \cup C_{13}$. Moreover,   $\omega[C_{12}] \leq 2$, $\omega[C_{13}] \leq 2$, $\omega[C_{23}] \leq 2$, by Claims 7 and 3. We colour $G$ with six colours as follows:
\begin{enumerate}[(1)]
\item Colour the vertices 1, 2, 3 of $A$ with colours 1, 2, 3 respectively.
\item Colour $[C_{12}]$ with colours 1 and 2.
\item Colour $[C_{23}]$ with colours 3 and 4.
\item Colour $[C_{13}]$ with colours 5 and 6.
\end{enumerate}
It is a proper colouring by the above observations.\\
\\Remarks:
\begin{enumerate}[(i)]
\item If some $C_{ij}$ is empty, we may not require all the six colours.
\item We do not have extremal graphs with chromatic number 6.
\item However, we do have a graph with chromatic number 4 (see Fig. \ref{fig:Examplew3X4}). In this figure, $A$ is an $\w$-clique and $N_i \subseteq V(G)$ such that every vertex of $N_i$ is adjacent to $i$ and only $i$ of $A$, $i \in \{1,2\}$.
\end{enumerate}
\vspace{0.6cm}

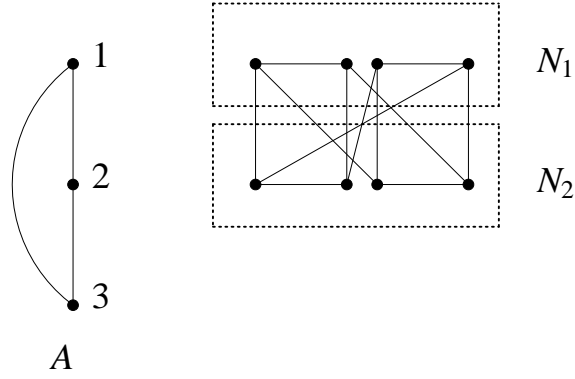
\begin{figure}
	\centering
		\scalebox{.8}{
		\begin{tikzpicture}[line cap=round,line join=round,>=triangle 45,x=1.0cm,y=1.0cm]
\clip(0.5,1.) rectangle (11.5,9.5);
\draw [shift={(4.5,6.)}] plot[domain=2.214297435588181:4.068887871591405,variable=\t]({1.*2.5*cos(\t r)+0.*2.5*sin(\t r)},{0.*2.5*cos(\t r)+1.*2.5*sin(\t r)});
\draw (3.,8.)-- (3.,6.);
\draw (3.,6.)-- (3.,4.);
\draw (6.,8.)-- (7.5,8.); 
\draw (6.,6.)-- (7.5,6.);
\draw (8,8.)-- (9.5,8.);
\draw (8,6)-- (9.5,6);
\draw (6.,8.)-- (6.,6.);
\draw (7.5,8.)-- (7.5,6.);
\draw (8,8.)-- (8,6.);
\draw (9.5,8.)-- (9.5,6.);
\draw (6.,8.)-- (8,6.); 
\draw (7.5,8.)-- (9.5,6.);
\draw (8,8.)-- (7.5,6.);
\draw (9.5,8)-- (6,6);

\draw [line width=.9pt,dotted] (5.3,9)-- (10.,9);
\draw [line width=.9pt,dotted] (5.3,7.3)-- (10.,7.3);
\draw [line width=.9pt,dotted] (5.3,9)-- (5.3,7.3);
\draw [line width=.9pt,dotted] (10.,7.3)-- (10.,9);

\draw [line width=.9pt,dotted] (5.3,7)-- (10.,7);
\draw [line width=.9pt,dotted] (5.3,5.3)-- (10.,5.3);
\draw [line width=.9pt,dotted] (5.3,7)-- (5.3,5.3);
\draw [line width=.9pt,dotted] (10.,5.3)-- (10.,7);

\draw (3.16,8.5) node[anchor=north west] {\Large$1$};
\draw (3.16,6.5) node[anchor=north west] {\Large$2$};
\draw (3.16,4.5) node[anchor=north west] {\Large$3$};
\draw (2.5,3.5) node[anchor=north west] {\Large$A$};
\draw (0.416363636363638,14.394545454545455) node[anchor=north west] {$1$};
\draw (10.49,8.41) node[anchor=north west] {\Large$N_{1}$};
\draw (10.49,6.41) node[anchor=north west] {\Large$N_{2}$};

\begin{scriptsize}

\draw [fill=black] (3.,4.) circle (2.5pt);
\draw [fill=black] (3.,6.) circle (2.5pt);
\draw [fill=black] (3.,8.) circle (2.5pt);

\draw [fill=black] (6.,8.) circle (2.5pt);
\draw [fill=black] (7.5,8.) circle (2.5pt);
\draw [fill=black] (8,8.) circle (2.5pt);
\draw [fill=black] (9.5,8.) circle (2.5pt);

\draw [fill=black] (6.,6.) circle (2.5pt);
\draw [fill=black] (7.5,6.) circle (2.5pt);
\draw [fill=black] (8,6.) circle (2.5pt);
\draw [fill=black] (9.5,6.) circle (2.5pt);

\end{scriptsize}
\end{tikzpicture}
		}
	\caption{$(P_3\cup P_2$, diamond)-free graph with $\w=3$ and $\chi=4$}
	\label{fig:Examplew3X4}
\end{figure}

\noindent \textbullet \quad   $\mathbf \omega = 4$; so $A = \{1,2,3,4\}$.\\
\\
We colour $G$ with five colours by considering two cases.\\
\\ \textbf{Case 1:} $[C_{23}$,  $C_{13}] \neq \emptyset$; let $ab \in [ C_{23},  C_{13}]$.
\\
Clearly, $[\{a,b,2\}] \simeq P_3$.\\
\noindent \\ \emph{Claim 9: $a$ is an isolated vertex in $[C_{23}]$, and $b$ is an isolated vertex in $[C_{13}]$.}\\
Suppose, $a \leftrightarrow c$, for some $c \in C_{23}$. If $c\leftrightarrow b$, then $[\{a,b,c,1\}] \simeq$ diamond, a contradiction. If $c\nleftrightarrow b$, then $[\{a,b,c\} \cup \{3,4\}] \simeq P_{3} \cup P_{2}$, since no vertex of $C_{23}\cup C_{13}$  is adjacent to the vertex $4 \in A$, by Claim 8. Hence, we conclude that $a$ is an isolated vertex in $C_{23}$. Similarly, $b$ is an isolated vertex in $C_{13}$.
\\
\\
\emph{Claim 10: $C_{23}$ and $C_{13}$ are independent sets.}\\
\noindent Suppose there exists an edge $cd$ in $[C_{23}]$, where $c \neq a$ and $d \neq a$, by Claim 9.  If $c \nleftrightarrow b$ and $d \nleftrightarrow b$, then $[\{a,b,2\} \cup \{c,d\}] \simeq  P_{3} \cup P_{2}$. Next, without loss of generality, suppose that $c \leftrightarrow b$.   Then  $[\{a,b,c \} \cup \{3,4\}] \simeq P_{3} \cup P_{2}$, by Claim 8 and by the definition of $C_{ij}$'s, a contradiction. Hence, $C_{23}$ is independent. Similarly $C_{13}$ is independent.\\
\\
We now colour $G$ with five colours as follows:
\begin{enumerate}[(1)]
\item Colour the vertices 1, 2, 3, 4 of $A$ with colours 1, 2, 3, 4 respectively.
\item Colour $[C_{12}]$ with colours 1, 2 and a new colour 5.
\item Colour $[C_{13}]$ with colour 3.
\item Colour $[C_{23}]$ with colour 4.
\end{enumerate}
It is a proper colouring  by Claims 8, 7 and 10.\\
\\ \textbf{Case 2:}  $[C_{23}, C_{13}] = \emptyset$.\\
If  both $C_{23}$ and $C_{13}$ are empty sets, then $G$ is $C_{5}$-free, since $[C_{12}]$ is $P_{3}$-free and any 5-cycle contains at most two vertices of $A$. So, $G$ is perfect, by Claim 4. If one of the sets $C_{23}$ or $C_{13}$ is nonempty, then we have the following assertion.\\
\\ \noindent \emph{Claim 11: If $C_{23}$ or $C_{13}$  is non empty, then the other is independent.}
\\
Suppose $C_{23} \neq \emptyset$ and $x\in C_{23}$. If $uv$ is an edge in $[C_{13}]$, then  $[\{x,1,3\} \cup \{u,v\}] \simeq  P_{3} \cup P_{2}$, a contradiction.  Hence $C_{13}$ is independent. Similarly, $C_{23}$ is independent if $C_{13} \neq \emptyset$.\\
\\Without loss of generality, we henceforth  assume that $C_{23} \neq \emptyset$. Since $C_{13}$ is nonempty or empty, we consider two subcases.\\
\\ \underline{Subcase 2.1:} $C_{13}$ is nonempty.
\\This implies that both $C_{23}$ and $C_{13}$ are independent sets, by Claim 11.
\begin{enumerate}[(1)]
\item Colour the vertices 1, 2, 3, 4 of $A$ with colours 1, 2, 3, 4 respectively.
\item Colour $[C_{12}]$ with colours 1, 2 and a new colour 5.
\item Colour $[C_{13}]$ with colour 3.
\item Colour $[C_{23}]$ with colour 3.
\end{enumerate}
It is a proper 5-colouring  by Claims 7, 11 and the fact that $[C_{23}, C_{13}] = \emptyset$.\\
\\
\underline{Subcase 2.2:} $C_{13}$ is empty.\\
We now examine this subcase based on number of components in $C_{23}$ and the maximum cliques in $C_{12}$.

\noindent Case 2.2.a: $C_{23}$ has exactly one component.\\
Recall that every component of $C_{23}$ is $K_{1}$, $K_{2}$ or $K_{3}$, by Claim 3. If the component is $K_1$, then  colour $G$ with five colours as follows:
\begin{enumerate}[(1)]
\item Colour the vertices 1, 2, 3, 4 of $A$ with colours 1, 2, 3, 4 respectively.
\item Colour $[C_{23}]$ with colour 3.
\item Colour $[C_{12}]$ with colours 1, 2 and a new colour 5.
\end{enumerate}
It is a proper 5-colouring by Claim 7 and by our assumptions.
\\
If the component is $K_2$ or $K_3$, let $cd$ be an edge in $[C_{23}]$ (see Fig. \ref{fig:OneComponentCase}). We claim that $C_{12}$ is independent. Else, there is an edge $ab$ in $[C_{12}]$. If $c$ is neither adjacent to $a$ nor adjacent to $b$, then $[\{c,1,2\}\cup\{a,b\}] \simeq P_{3}\cup P_{2}$, a contradiction. Without loss of generality, assume that $a\leftrightarrow c$. But then $a \nleftrightarrow d$; else, $[\{a,c,d,1\}] \simeq$ diamond. By definition of $C_{12}$ and $C_{23}$, no vertex in $\{a,c,d\}$ is adjacent to vertex 2 of $A$. By Claim 8, $a$ is adjacent to at most one vertex of $A -\{1,2\}$, namely 3 or 4. So $[\{a,c,d\}\cup\{2,3\}] \simeq P_{3}\cup P_{2}$ or $[\{a,c,d\}\cup\{2,4\}] \simeq P_{3}\cup P_{2}$, a contradiction. Hence, $C_{12}$ is independent. Recall that $\omega ([C_{23}]) \leq 3$, by Claim 3.\\
We colour G with four colours:
\begin{enumerate}[(1)]
\item Colour the vertices 1, 2, 3, 4 of $A$ with colours 1, 2, 3, 4 respectively.
\item Colour $C_{23}$ with colours 2, 3 and 4.
\item Colour $C_{12}$ with colour 1.
\end{enumerate}
It is a proper 4-colouring by Claims 3 and 8.

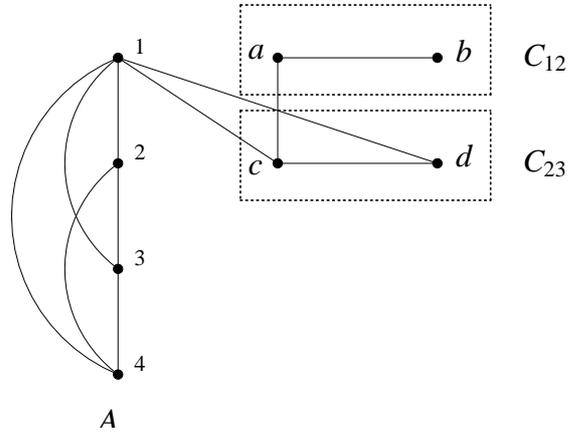
\begin{figure}[H]
\centering
\scalebox{0.7}{
\begin{tikzpicture}[line cap=round,line join=round,>=triangle 45,x=1.0cm,y=1.0cm]
\clip(0.5,1.) rectangle (11.5,9.5);
\draw [shift={(4.5,6.)}] plot[domain=2.214297435588181:4.068887871591405,variable=\t]({1.*2.5*cos(\t r)+0.*2.5*sin(\t r)},{0.*2.5*cos(\t r)+1.*2.5*sin(\t r)});
\draw [shift={(4.5,4.)}] plot[domain=2.214297435588181:4.068887871591405,variable=\t]({1.*2.5*cos(\t r)+0.*2.5*sin(\t r)},{0.*2.5*cos(\t r)+1.*2.5*sin(\t r)});
\draw [shift={(4.25,5.)}] plot[domain=1.965587446494658:4.317597860684928,variable=\t]({1.*3.25*cos(\t r)+0.*3.25*sin(\t r)},{0.*3.25*cos(\t r)+1.*3.25*sin(\t r)});
\draw (3.,8.)-- (3.,6.);
\draw (3.,6.)-- (3.,4.);
\draw (3.,4.)-- (3.,2.);
\draw (6.,8.)-- (9.,8.); 
\draw (6.,6.)-- (9.,6.);
\draw (3.,8.)-- (6.,6.);
\draw (3.,8.)-- (9.,6.);

\draw [line width=.9pt,dotted] (5.3,9)-- (10.,9);
\draw [line width=.9pt,dotted] (5.3,7.3)-- (10.,7.3);
\draw [line width=.9pt,dotted] (5.3,9)-- (5.3,7.3);
\draw [line width=.9pt,dotted] (10.,7.3)-- (10.,9);

\draw [line width=.9pt,dotted] (5.3,7)-- (10.,7);
\draw [line width=.9pt,dotted] (5.3,5.3)-- (10.,5.3);
\draw [line width=.9pt,dotted] (5.3,7)-- (5.3,5.3);
\draw [line width=.9pt,dotted] (10.,5.3)-- (10.,7);

\draw (3.16,8.5) node[anchor=north west] {$1$};
\draw (3.16,6.5) node[anchor=north west] {$2$};
\draw (3.16,4.5) node[anchor=north west] {$3$};
\draw (3.16,2.5) node[anchor=north west] {$4$};
\draw (2.5,1.5) node[anchor=north west] {\Large$A$};
\draw (0.416363636363638,14.394545454545455) node[anchor=north west] {$1$};
\draw (5.3,8.4) node[anchor=north west] {\Large$a$};
\draw (9.2,8.5) node[anchor=north west] {\Large$b$};
\draw (5.3,6.2) node[anchor=north west] {\Large$c$};
\draw (9.2,6.5) node[anchor=north west] {\Large$d$};
\draw (10.49,8.41) node[anchor=north west] {\Large$C_{12}$};
\draw (10.49,6.41) node[anchor=north west] {\Large$C_{23}$};
\draw (6.,8.)-- (6.,6.);
\begin{scriptsize}
\draw [fill=black] (3.,2.) circle (2.5pt);
\draw [fill=black] (3.,4.) circle (2.5pt);
\draw [fill=black] (3.,6.) circle (2.5pt);
\draw [fill=black] (3.,8.) circle (2.5pt);
\draw [fill=black] (6.,8.) circle (2.5pt);
\draw [fill=black] (9.,8.) circle (2.5pt);
\draw [fill=black] (6.,6.) circle (2.5pt);
\draw [fill=black] (9.,6.) circle (2.5pt);
\end{scriptsize}
\end{tikzpicture}
}
\caption{$[C_{23}]$ has one component}
	\label{fig:OneComponentCase}
\end{figure}
\noindent Case 2.2.b: $C_{23}$ has $\geq 2$ components; let $x$ and $y$ be vertices of two distinct components (see Fig. \ref{fig:TwoComponentCase}).\\
Our first claim is that $\omega ([C_{12}]) \leq 2$. On the contrary suppose that $[\{a,b,c\}]$ is a triangle in $[C_{12}]$. Since $\{x,1,2\}$ induces a $P_{3}$, $x$ is adjacent to every vertex of the triangle; else we have an induced diamond or $P_{3}\cup P_{2}$ in $G$. Similarly $y$ is adjacent to every vertex of the triangle. Then $[\{a,b,x,y\}] \simeq$ diamond. Hence, $\omega([C_{12}]) \leq 2$. So we can colour $G$ with 4 colours as follows:
\begin{enumerate}[(1)]
\item Colour the vertices 1, 2, 3, 4 of $A$ with colours 1, 2, 3, 4 respectively.
\item Colour $C_{23}$ with colours 3 and 4.
\item Colour $C_{12}$ with colour 1 and 2.
\end{enumerate}
It is a proper 4-colouring by the above observations and Claim 8.

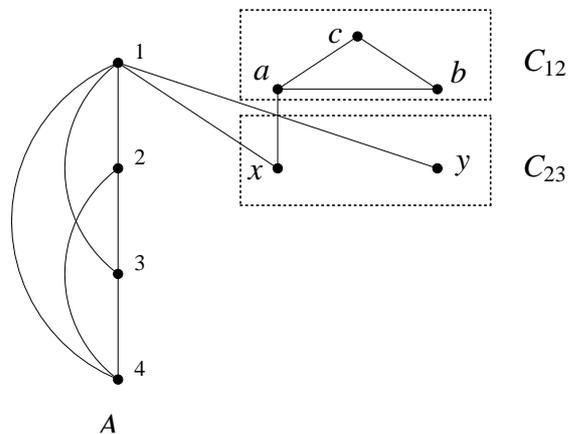
\begin{figure}[H]
	\centering
\scalebox{.7}{
\begin{tikzpicture}[line cap=round,line join=round,>=triangle 45,x=1.0cm,y=1.0cm]
\clip(0.5,1.) rectangle (11.5,9.5);
\draw [shift={(4.5,6.)}] plot[domain=2.214297435588181:4.068887871591405,variable=\t]({1.*2.5*cos(\t r)+0.*2.5*sin(\t r)},{0.*2.5*cos(\t r)+1.*2.5*sin(\t r)});
\draw [shift={(4.5,4.)}] plot[domain=2.214297435588181:4.068887871591405,variable=\t]({1.*2.5*cos(\t r)+0.*2.5*sin(\t r)},{0.*2.5*cos(\t r)+1.*2.5*sin(\t r)});
\draw [shift={(4.25,5.)}] plot[domain=1.965587446494658:4.317597860684928,variable=\t]({1.*3.25*cos(\t r)+0.*3.25*sin(\t r)},{0.*3.25*cos(\t r)+1.*3.25*sin(\t r)});
\draw (3.,8.)-- (3.,6.);
\draw (3.,6.)-- (3.,4.);
\draw (3.,4.)-- (3.,2.);
\draw (6.,7.5)-- (9.,7.5); 
\draw (6.,7.5)-- (7.5,8.5); 
\draw (7.5,8.5)-- (9.,7.5); 
\draw (3.,8.)-- (6.,6.);
\draw (3.,8.)-- (9.,6.);
\draw (6.,7.5)-- (6.,6.); 

\draw [line width=.9pt,dotted] (5.3,9)-- (10.,9);
\draw [line width=.9pt,dotted] (5.3,7.3)-- (10.,7.3);
\draw [line width=.9pt,dotted] (5.3,9)-- (5.3,7.3);
\draw [line width=.9pt,dotted] (10.,7.3)-- (10.,9);

\draw [line width=.9pt,dotted] (5.3,7)-- (10.,7);
\draw [line width=.9pt,dotted] (5.3,5.3)-- (10.,5.3);
\draw [line width=.9pt,dotted] (5.3,7)-- (5.3,5.3);
\draw [line width=.9pt,dotted] (10.,5.3)-- (10.,7);

\draw (3.16,8.5) node[anchor=north west] {$1$};
\draw (3.16,6.5) node[anchor=north west] {$2$};
\draw (3.16,4.5) node[anchor=north west] {$3$};
\draw (3.16,2.5) node[anchor=north west] {$4$};

\draw (5.4,8.1) node[anchor=north west] {\Large$a$};
\draw (9.1,8.2) node[anchor=north west] {\Large$b$};
\draw (6.8,8.8) node[anchor=north west] {\Large$c$};
\draw (5.3,6.2) node[anchor=north west] {\Large$x$};
\draw (9.2,6.4) node[anchor=north west] {\Large$y$};
\draw (10.49,8.41) node[anchor=north west] {\Large$C_{12}$};
\draw (10.49,6.41) node[anchor=north west] {\Large$C_{23}$};
\draw (2.5,1.5) node[anchor=north west] {\Large$A$};
\begin{scriptsize}
\draw [fill=black] (3.,2.) circle (2.5pt);
\draw [fill=black] (3.,4.) circle (2.5pt);
\draw [fill=black] (3.,6.) circle (2.5pt);
\draw [fill=black] (3.,8.) circle (2.5pt);
\draw [fill=black] (6.,7.5) circle (2.5pt);
\draw [fill=black] (9.,7.5) circle (2.5pt);
\draw [fill=black] (7.5,8.5) circle (2.5pt);
\draw [fill=black] (6.,6.) circle (2.5pt);
\draw [fill=black] (9.,6.) circle (2.5pt);
\end{scriptsize}
\end{tikzpicture}
}
	\caption{$[C_{23}]$ has more than one component}
	\label{fig:TwoComponentCase}
\end{figure}

\noindent \textbullet \quad   $\mathbf \omega \geq 5$.\\
\\It is enough to show that $G$ is $C_{5}$-free, in view of Claim 4.4. On the contrary, suppose that \emph{G} contains an induced  $C_5$. As before, $V(G)- A = C = C_{12} \cup C_{13} \cup C_{23}$.  Since at most two  vertices of $C_5$ can belong to the clique \emph{A}, a $P_{3}=(a,b,c)$ is an induced subgraph of $[C]$. Since each $C_{ij}$ is $P_3$-free, either ($i$) two vertices are in one $C_{ij}$, and the third vertex is in one of the other two $C_{ij}$'s, or ($ii$) each $C_{ij}$ contains a vertex. 
\\
\\
\noindent \emph{Claim 12: A vertex of $C_{12}$ is adjacent  to at most one vertex of $A$.}\\
\noindent The claim is obvious for $\w=2,3$. Next, assume that $\w \geq 4$. If some vertex $x \in C_{12}$ is adjacent to two distinct vertices say, \emph{i} and \emph{j} of $A-\{1,2\}$, then $[\{1, x, i, j\}]  \simeq$ diamond, a contradiction.
\\
\\
Hence by the above claim, for any two vertices $x,y \in C_{12}$, there is a vertex, say 5, in $A$ which is neither adjacent to $x$ nor $y$. Also, by Claim 8, $[C_{13} \cup C_{23}, \{3,4,5\}] = \emptyset$. So, whether ($i$) or ($ii$) holds, there exists an edge $ij$ in $[A]$ such that $[\{a,b,c\}\cup\{i,j\}] \simeq P_{3} \cup P_2$, a contradiction. For the choice of an appropriate edge $ij$, it is enough if we consider the following four cases:
\begin{enumerate}[(a)]
\item If  $P_{3}$ is an induced subgraph of $[\{C_{12}\cup C_{13}\}]$, then $[\{a,b,c,1,5\}] \simeq P_{3}\cup P_{2}$.
\item If  $P_{3}$ is an induced subgraph of $[\{C_{12}\cup C_{23}\}]$, then $[\{a,b,c,2,5\}] \simeq P_{3}\cup P_{2}$.
\item If  $P_{3}$ is an induced subgraph of $[\{C_{13}\cup C_{23}\}]$, then $[\{a,b,c,4,5\}] \simeq P_{3}\cup P_{2}$.
\item If ($ii$) holds, then $[\{a,b,c,4,5\}] \simeq P_{3}\cup P_{2}$, where without loss of generality we assume that the vertex of $(a,b,c)$ that is in $C_{12}$ is adjacent to the vertex $3 \in A$.
\end{enumerate}
\end{proof}

\section{\texorpdfstring{$(2K_2,$ diamond)}{Lg}-free graphs}
The Claims of Section 4 are valid for $(2K_2,$ diamond)-free graphs too. So  we continue to use the Claims made in Sections 3 and 4. In what follows, we assume that graphs have clique number at least 2, as before.\\

\begin{theorem} \label{thm:2K2,diamond}
If a graph G is $(2K_{2},$ diamond)-free, then
\begin{center}
\begin{math}
\chi(G) \leq \begin{cases} \w+1 & \mbox{if } \w=2 \\
\w & \mbox{if } \w \geq 3\\
\end{cases}
\end{math}
\end{center}
and $G$ is perfect if $\w \geq 4$.
\end{theorem}

\begin{proof}
Since the proof is similar to the proof of Theorem \ref{thm:P3UP2,diamond}, we give an outline.
As before, consider the partition $(A, \bigcup C_{ij}, \bigcup I_{a})$ of $V(G)$. In this case, every $C_{ij}$ is $K_{2}$-free, and so it is an independent set.
\\
\noindent If $\w=2$, then $V(G)= A \cup C_{12} \cup I_{1} \cup I_{2}$. So one can easily colour $G$ with three colours.
\\
\noindent Next suppose $\w \geq 3.$ If $j \in A$, then $I_{j} = \emptyset$. Else, some $x \in I_{j}$. So, if $a,b \in A - \{j\}$, then $[\{x,j, a,b\}] \simeq$ diamond, a contradiction.
Also, $C_{ij} = \emptyset$, if $j \geq 4$; else $G$ contains an induced diamond. Hence $V(G)= C_{12} \cup C_{13}\cup C_{23}$. An $\omega$-colouring of $G$ is obtained as follows:
\begin{enumerate}[(1)]
\item Colour the vertices $1,2,\cdots ,\omega$ of $A$, by colours $1,2,\cdots ,\omega$.
\item Colour every vertex of $C_{12}$ with colour 1, colour every vertex of $C_{13}$ with colour 3, colour every vertex of $C_{23}$ with colour 2.
\end{enumerate}

Remark: There exist $(2K_{2},$ diamond)-free graphs with $\omega = 3$, which are not perfect. See Fig. \ref{fig:nonperfect}, where each circled vertex is multiplied by an independent set.

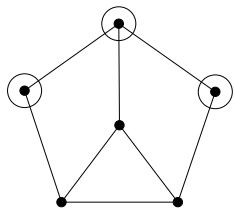
\begin{figure}[H]
\centering
  \scalebox{.7}{
  \begin{tikzpicture}[line cap=round,line join=round,>=triangle 45,x=1.0cm,y=1.0cm]
\clip(2.5,-0.5) rectangle (7.5,3.7);
\draw (3.9102733048340688,-0.13699305960311253)-- (3.217010452158673,1.974792245469635);
\draw (3.217010452158673,1.974792245469635)-- (4.987497122068145,3.243996544983054);
\draw (4.987497122068145,3.243996544983054)-- (6.800646121373028,1.9534610807719304);
\draw (6.800646121373028,1.9534610807719304)-- (6.096717686348779,-0.13699305960311253);
\draw (6.096717686348779,-0.13699305960311253)-- (3.9102733048340688,-0.13699305960311253);
\draw (3.9102733048340688,-0.13699305960311253)-- (5.,1.32);
\draw (5.,1.32)-- (6.096717686348779,-0.13699305960311253);
\draw (5.,1.32)-- (4.987497122068145,3.243996544983054);
\draw(4.987497122068145,3.243996544983054) circle (0.3199674704655678cm);
\draw(6.800646121373028,1.9534610807719304) circle (0.3199674704655678cm);
\draw(3.217010452158673,1.974792245469635) circle (0.3199674704655678cm);
\begin{scriptsize}
\draw [fill=black] (5.,1.32) circle (2.5pt);
\draw [fill=black] (3.9102733048340688,-0.13699305960311253) circle (2.5pt);
\draw [fill=black] (6.096717686348779,-0.13699305960311253) circle (2.5pt);
\draw [fill=black] (6.800646121373028,1.9534610807719304) circle (2.5pt);
\draw [fill=black] (3.217010452158673,1.974792245469635) circle (2.5pt);
\draw [fill=black] (4.987497122068145,3.243996544983054) circle (2.5pt);
\end{scriptsize}
\end{tikzpicture}}
\caption{Graphs that are not perfect and have $\chi(G)=\omega(G)$ }{}
\label{fig:nonperfect}
\end{figure}

Now we prove perfectness for $\w \geq 4.$
\\It is similar to the proof of Theorem \ref{thm:P3UP2,diamond}, Case $\omega = 5$. By Claim 4.4, it is enough if we show that $G$ is $C_{5}$-free. On the contrary, if $G$ contains an induced 5-cycle, then $C (= C_{12} \cup C_{13} \cup C_{23})$ contains an edge $xy$ of the 5-cycle. Since $C_{ij}$'s are independent, no $[C_{ij}]$ contains $xy$. We use Claims 8 and 12 and arrive at a contradiction:
\begin{enumerate}[(a)]
\item If $xy \in [C_{12}, C_{13}]$, then $[\{x,y,1,3\}] = 2K_{2}$ or $[\{x,y,1,4\}] = 2K_{2}$.
\item If $xy \in [C_{12}, C_{23}]$, then $[\{x,y,2,3\}] = 2K_{2}$ or $[\{x,y,2,4\}] = 2K_{2}$.
\item If $xy \in [C_{13}, C_{23}]$, then $[\{x,y,1,3\}] = 2K_{2}$ or $[\{x,y,1,4\}] = 2K_{2}$.
\end{enumerate}

So, $G$ is $C_{5}$-free and hence it is perfect.
\end{proof}

\section*{Acknowledgements}
Both the authors thank Christ University, Bengaluru for providing all the facilities to do this research.








\begin{thebibliography}{00}
\bibitem{ENS}
Nagy, Zs. and Szentmiklossy, Z.: A \$20 open problem of Erd\"{o}s to show that if $G$ is $2K_{2}$-free graph with clique number 3, then its chromatic number is 4 (Unpublished)
\end{thebibliography}







\end{document}